\documentclass[12pt]{article}
\usepackage[sort&compress,square,comma,numbers]{natbib}

\usepackage{hyperref}
\hypersetup{colorlinks=true,urlcolor=blue,linkcolor=blue,citecolor=blue}

\usepackage{amsmath}
\usepackage{amsthm}
\usepackage{amssymb}
\usepackage{amsmath}
\usepackage{vmargin}

\usepackage{arvin}
\usepackage{nickstyle_subset}

\setpapersize[portrait]{USletter}

\setmarginsrb{0.9in}{0.6in}{0.9in}{0.7in}{12pt}{12pt}{12pt}{16pt}

\title{Explicit Orthogonal Arrays and Universal Hashing \\ with Arbitrary Parameters}

\author{Nicholas Harvey \qquad Arvin Sahami}
\date{}

\newcommand{\SU}[1]{$#1$-independent}
\newcommand{\Cnstr}{\textsc{Constructor}\xspace}
\newcommand{\Eval}{\textsc{Eval}\xspace}
\newcommand{\Hash}{\textsc{Hash}\xspace}
\newcommand{\CnstrNoSpace}{\textsc{Constructor}}
\newcommand{\EvalNoSpace}{\textsc{Eval}}
\newcommand{\HashNoSpace}{\textsc{Hash}}

\newcommand{\Uniform}{\textsc{Uniform}}
\newcommand{\Special}{\beta}
\newcommand{\poly}{\operatorname{poly}}

\begin{document}

\maketitle

\begin{abstract}
Orthogonal arrays are a type of combinatorial design that were developed in the 1940s in the design of statistical experiments.
In 1947, Rao proved a lower bound on the size of any orthogonal array, and raised the problem of constructing arrays of minimum size.
Kuperberg, Lovett and Peled (2017) gave a non-constructive existence proof of orthogonal arrays whose size is near-optimal (i.e., within a polynomial of Rao's lower bound), leaving open the question of an algorithmic construction.
We give the first explicit, deterministic, algorithmic construction of orthogonal arrays achieving near-optimal size for all parameters. 
Our construction uses algebraic geometry codes.

In pseudorandomness, the notions of $t$-independent generators or $t$-independent hash functions are equivalent to orthogonal arrays.
Classical constructions of \SU{t} hash functions are known when the size of the codomain is a prime power, but very few constructions are known for an arbitrary codomain.
Our construction yields algorithmically efficient \SU{t} hash functions for arbitrary domain and codomain.
\end{abstract}

\section{Introduction}

Orthogonal arrays are a concept in the design of statistical experiments, first proposed by C.~R.~Rao in the 1940s.
A detailed exposition of this subject can be found in reference books \cite{Hedayat} \citep{StinsonBook}.
An orthogonal array with parameters $[s,m,n,t]$ is a matrix with $s$ rows, $m$ columns and entries in $[n]=\set{1,\ldots,n}$ such that, for every set of $t$ columns, those columns contain each tuple in $[n]^t$ exactly $s/n^t$ times.
A long history of research considers the problem of, given $m$, $n$ and $t$, finding an orthogonal array with small $s$.

If $n$ is a prime power, a classical construction due to Bush\footnote{This is essentially equivalent to a Reed-Solomon code, and predates it by eight years.}
in 1952 \cite[Theorem 3.1]{Hedayat} \cite[\S 3]{BushPolynomials} gives an orthogonal array with exactly $s = n^t$ rows if $m \leq n$.
This is exactly optimal (the definition requires that $s/n^t$ is a positive integer). 
For $m > n$, a simple modification of this construction has $s \leq (nm)^t$ rows; however this is no longer optimal.

Indeed, a lower bound of Rao\footnote{
For the case $n=2$, this can be found in \cite[Section 8]{Chor}.
In the special case of linear orthogonal arrays, Rao's bound is equivalent to the dual of the generalized Hamming bound on codes, and predates it by three years.
}
\cite[Theorem 2.1]{Hedayat} implies that every orthogonal array has
\begin{equation}
\EquationName{RaoBound}
s ~\geq~ \binom{m}{\floor{t/2}} (n-1)^{\floor{t/2}}
 ~\geq~ (cmn/t)^{c t},
\end{equation}
for some constant $c>0$.
In contrast, Bush's upper bound lacks the denominator of $t$, and only applies when $n$ is a prime power.
In general, there are very few constructions of orthogonal arrays when $n$ is not a prime power; 
see \cite[Table 12.6]{Hedayat}.
\begin{quote}
Some important problems are the constructions of these arrays with the optimum values of $m$ for a given $t$ and $s$. \\
\mbox{}\hspace{1cm} ~Rao, 1947 \cite{Rao1947}

\vspace{6pt}

\textbf{Research Problem 2.33.} For fixed values of $n$ and $t$, and large $m$, how far are
the Rao bounds from the truth? \\
\mbox{}\hspace{1cm}
Hedayat, Sloane, Stufken 1999
\cite{Hedayat}

{\footnotesize
\textit{(Note: notation has been adjusted to match ours.)}}
\end{quote}

\paragraph{Previous constructions.}
One might imagine that a simple product construction can reduce the general case to the prime power case.
This is possible (and was known to Bose in 1950 \cite[Section 2]{Bush}) but it is suboptimal.
First factor $n$ into a product of prime powers $\prod_{\ell \leq d} p_\ell^{a_\ell}$.
For each $\ell$, build an orthogonal array $A^{(\ell)}$ which in which $n$ is taken to be $p_\ell^{a_\ell}$.
Then, define $A$ to be an entry-wise Cartesian product of $A^{(1)},\ldots,A^{(d)}$.
This gives an orthogonal array, but it is unfortunately somewhat large.
Applying Rao's lower bound separately to each $A^{(\ell)}$, the number of rows of $A$ is at least
$
\prod_{\ell \leq d} (c m p_\ell^{a_\ell}/t)^{ct} 
= (cm/t)^{c d t} n^{ct}.
$
This does not asymptotically match Rao's lower bound for all parameters due to the factor $d$ in the exponent.
So this construction cannot in general prove that Rao's lower bound is tight.
Note that $d$ can be $\Omega(\log(n)/\log \log n)$ when $n$ is the product of the first $d$ primes.

If $n$ is a prime power, then it is known how to explicitly construct orthogonal arrays with $s \leq (C m n / t)^{Ct}$ for some constant $C$.
We will call such a construction \textit{near-optimal}, meaning that it matches Rao's lower bound up to the value of the constant $C$.
The anonymous reviewers of this manuscript have informed us of this result, which is stated in \cite[Theorem 3.4]{CL21} using the language of pseudorandom generators.
That construction, like our results in \Section{AG}, is based on the use of algebraic geometry codes. 

If $n$ is not restricted to be a prime power, then
Kuperberg, Lovett and Peled~\cite{KLP17} give a \emph{non-explicit} construction, for all $m,n,t$, of orthogonal arrays that are also near-optimal, so $s \leq (Cmn/t)^{Ct}$ for some constant $C$.
However their approach is randomized, and they only prove an exponentially small lower bound on the probability of success. 
One of their open questions is whether an algorithmic version of their construction exists \cite[pp.~968]{KLP17}.

\paragraph{Our results.}
We also give a \textit{near-optimal} construction of orthogonal arrays.
Our construction is explicit, algorithmic, works for all $m, n, t$, and is apparently unrelated to the construction of Kuperberg, Lovett and Peled.

\begin{theorem}[Informal]
\TheoremName{Main}
For any $m,n,t$ (with $n$ not necessarily a prime power), there is an explicit description of an orthogonal array with parameters $[s,m,n,t]$ where $s=(cmn/t)^{35t}$, and $c$ is a universal constant.
The array can be constructed by a deterministic algorithm with runtime $\poly(sm)$.
\end{theorem}

To prove this, we give an abstract construction based on coding theory.
In \Section{RS} we instantiate that construction with Reed-Solomon codes, obtaining a simple, easily implementable, deterministic construction with $s = O(m+n)^{6t}$.
In \Section{AG} we instantiate the construction instead with algebraic geometry codes to obtain a deterministic, near-optimal construction with $s = O(mn/t)^{35t}$.
In \Section{GV} we instantiate the construction instead with random linear codes to obtain a non-explicit construction that has near-optimal size, and can be efficiently constructed by a randomized algorithm.

\subsection{Hash Functions}

Hash functions of various sorts are crucial tools in pseudorandomness and randomized algorithms.
In this work we focus on \emph{\SU{t}} functions.
Let $h : [m] \rightarrow [n]$ be a random function with some distribution.
We say that $h : [m] \rightarrow [n]$ is
\SU{t} if
\begin{align*}
&\prob{ h(a_1)=\alpha_1 \wedge\cdots
\wedge h(a_t)=\alpha_t } ~=~ \frac{1}{n^t} \\
&\text{$\forall$~distinct~} a_1,\ldots,a_t \in [m],~\forall \alpha_1,\ldots,\alpha_t \in [n].
\end{align*}
(Obviously $t \leq m$ is necessary.)
This is equivalent to the random variables $\setst{ h(a) }{ a \in [m] }$ being uniformly distributed and $t$-wise independent.
If this property is satisfied by a function $h$ that is uniformly chosen from a multiset $\cH$, then we also say that $\cH$ is \SU{t}.
This property is also called \emph{\ensuremath{\text{strongly universal}_t}}, or simply \emph{strongly universal} if $t=2$.

Wegman and Carter \cite[Section 1]{WC81} gave a simple construction of a \SU{t} family $\cH$ when \textit{$n$ is a prime power}\footnote{
Carter and Wegman also defined the notion of \emph{weakly} universal hash functions, and gave a construction that allows $n$ to be an arbitrary integer \cite[Proposition 7]{CW79}.
} and $m=n$.
In this construction, $\cH$ corresponds to the polynomials of degree $t-1$ over the field $\bF_n$, so $\card{\cH}=n^t$ and the space required to represent a member of $\cH$ is $O(t \log n)$ bits.
In general, if $n$ \emph{is a prime power} and $m$ is arbitrary, then a modified construction has $\card{\cH} = (mn)^{O(t)}$, so members of $\cH$ can be represented in $O(t \log(mn))$ bits.

Orthogonal arrays and \SU{t} hash families are
mathematically equivalent notions, as was observed by Stinson\footnote{
Another connection between design theory and derandomization is in Karp and Wigderson \cite[\S 9]{karp1985fast}.}
\cite[Theorem 3.1]{Stinson94} \cite[Theorem
5.2]{StinsonSurvey}.

\begin{claim}
\ClaimName{Equiv}
Let $M$ be an $s \times m$ matrix with entries in $[n]$.
Let $\cH = \set{ h_1,\ldots,h_s }$ be a multiset of
$[m] \rightarrow [n]$ functions where $h_i(j) = M_{i,j} ~\forall j \in [m]$.
Then $\cH$ is \SU{t} iff $M$ is a $[s,m,n,t]$ orthogonal array.
\end{claim}

The proof is immediate from the definitions.
Thus Wegman and Carter's construction of \SU{t} hash functions is identical to Bush's 1952 construction of orthogonal arrays.

There are some differences in the goals of the research communities working on orthogonal arrays and \SU{t} hashing.
One difference is computational.
Researchers in orthogonal arrays tend to be interested in showing existence, or explicit construction, of the entire matrix $M$. 
In contrast, researchers in hash functions are additionally interested in quickly sampling a hash function from $\cH$ and quickly evaluating the hash function.
The following theorem states our main result for hash functions.

\begin{theorem}[Informal]
\TheoremName{MainHash}
For any parameters $m,n,t$, there exists a \SU{t} hash family for which functions can be represented in $O(t \log(mn))$ bits, can be evaluated in time $O(t)$, and can be constructed in expected time $O(t + (n+m)^\eps)$ for any $\eps>0$. 
Assuming the generalized Riemann hypothesis (GRH), the expected construction time improves to $O(t) + \polylog(nm)$.
\end{theorem}

In contrast, the Wegman-Carter construction requires that $n$ \emph{is a prime power}, also uses $O(t \log(nm))$ bits, also can be evaluated in time $O(t)$, and can be constructed in expected time $O(t) + \polylog(nm)$.
(If $m > n$, the construction must work in a field extension of $\bF_n$ of size at least $m$, so random sampling is used to find an irreducible polynomial.)   
Thus the efficiency of our hash function matches theirs, assuming GRH.

\subsection{Applications}

There are many algorithms that require \SU{t} hash functions for $t>2$. Examples include 
cuckoo hashing \cite{Kane},
distinct element estimation \cite{KNW}, MinHash \cite{Indyk},
and the leftover hash lemma, which is used to construct seeded extractors \cite{StinsonLeftover}, etc.
More recently, a mechanism for maximizing Nash social welfare in the percentage fee model uses \SU{t} hash functions where $n$ not a prime power \cite{dobzinski2023fairness}.

\section{The abstract construction using codes}
\SectionName{GenericConstruction}

In this section we present an abstract construction of orthogonal arrays, even when $n$ is not a prime power, using linear codes over finite fields.
The subsequent sections instantiate this construction using particular codes.
The Hamming distance between vectors $x$ and $y$ is denoted $\Delta(x,y) = \card{\setst{ i }{ x_i \neq y_i }}$.

\begin{lemma}
\LemmaName{BuildOAWorks}
Let $m,n,t$ be integers with $n \geq 2$ and $2 \leq t \leq m$.
Let $q$ be a prime power satisfying $q \equiv 1 \pmod n$.
Suppose that $C \subseteq \bF_q^{m}$ is a linear code of dimension $k$ whose dual has minimum distance at least $t+1$.
Let $b \in \bF_q^{m}$ satisfy $\Delta(b,u) \geq m-\tau$ for all $u \in C$.
Then \textsc{BuildOA}, shown in \Algorithm{BuildOA2}, will produce an $[s,m,n,t]$ orthogonal array with $ s = n^\tau q^k$.
\end{lemma}

An intuitive explanation of this lemma is as follows.
If $u \in C$ were chosen uniformly at random, our hypotheses on $C$ would imply that the coordinates of $u$ will be uniformly distributed and $t$-wise independent.
Thus, if for each $u \in C$, we insert $u$ as a new row in the array $M$, then the columns of $M$ would also be uniformly distributed and $t$-wise independent, which is equivalent to $M$ being an orthogonal array.
The catch, of course, is that the entries of $u$ take values in $\bF_q$, whereas the entries of $M$ must take values in $[n]$.
If $q$ were a multiple of $n$, then this problem would be easily resolved: we could instead insert $u$ modulo $n$ (entry-wise), as a new row in $M$.
The resulting matrix $M$ is easily seen to be an orthogonal array.
More generally, the same construction would work using any $\frac{q}{n}$-to-$1$ map $f : \bF_q \rightarrow [n]$
applied entry-wise to $u$ to create a row of $M$.

However, the main scenario of interest is the one in which $n$ is a composite number,
so $q$ (being a prime power) certainly cannot be a multiple of $n$.
Instead, we have $q \equiv 1 \pmod n$, so we require a new approach to create rows of $M$ from the vectors $u \in C$.
The idea is to choose a ``bad value'' $\beta$ to eliminate from $\bF_q$, so that the number of remaining ``good values'' is a multiple of $n$, and these good values can be mapped to integers in $[n]$ in a way that preserves the uniform distribution.
This mapping, which we denote $\phi_\beta$, between good values and $[n]$ can be completely arbitrary, except that it needs to be a $\frac{q-1}{n}$-to-1 map in order to ensure uniformity.
For example, if we identify the elements in $\bF_q$ with the integers in $\set{0,\ldots,q-1}$, then one may check that
\begin{equation}
\EquationName{PhiDef}
\phi_\Special(x) = 1+\Big(\big((x+q-1-\Special) \text{~mod~} q \big) \text{~mod~} n \Big)
\end{equation}
is indeed a $\frac{q-1}{n}$-to-1 map from $\bF_q \setminus \set{\beta}$ to $[n]$.

For future purposes, it will be convenient to allow this bad value to depend on the coordinate.
Using the given vector $b$, we let $b_j$ be the bad value for coordinate $j$.
To eliminate these bad values from $u$, 
we must find all indices $j$ with $u_j=b_j$ and remap $u_j$ to a new value.
Intuitively, we would like to do so uniformly at random.
However the construction must be deterministic, so instead we create multiple copies of $u$ in which each such $u_j$ has been replaced with every possible value in $[n]$.
In \Algorithm{BuildOA2}, the vector $v$ specifies how the bad entries of $u$ will be fixed, whereas the vector $w(u,v)$ is the modified copy of $u$ that incorporates all the fixes and maps all entries to $[n]$.
To control the number of copies that are created, we must bound the number of bad indices.
This is ensured by the condition $\Delta(b,u) \geq m - \tau$: 
every $u \in C$ has at most $\tau$ bad indices.
Lastly, to ensure that every two codewords in $C$ generate the \emph{same number} of rows in $M$, we add $n^{\tau-\ell}$ copies of each modified vector $w(u,v)$, where $\ell=m-\Delta(b,u)$ is the number of indices at which $u$ and $b$ agree.

\begin{algorithm}[t]
\caption{\textsc{BuildOA} uses the code $C$ and vector $b$ to construct an orthogonal array $M$.
Here $\phi_\beta$ is an arbitrary $\frac{q-1}{n}$-to-1 map from $\bF_q\setminus\set{\beta}$ to $[n]$,
as in \eqref{eq:PhiDef}.
}
\AlgorithmName{BuildOA2}
\begin{algorithmic}[1]
        \Function{\textsc{BuildOA}}{integer $m, n, t, q, \tau$, code $C$, vector $b$}
 	\State Let $M$ be a matrix with $m$ columns and no rows
	\For{$u \in C$}
		\State Let $Z= \set{z_1 ,\ldots , z_\ell} \subseteq [m]$ \\
        \makebox[25pt]{} be the set of indices with $u_{z_i}=b_{z_i}$
		\State \textbf{for} $v \in [n] ^ \ell$ \textbf{do}
            \State \hspace{16pt} $\rhd$ $v$ specifies how to fix all bad values in $u$
			\State \hspace{16pt} Define the vector $w(u,v) \in [n]^m$ by 
			\State 
				\[
					\hspace{1.5cm} w(u,v)_j = 
					\begin{cases} 
						v_i 		& \text{if } j = z_i \text{ for some } i \in [\ell] \\
						\phi_{b_j}(u_j) 	& \text{otherwise, i.e., if } u_j \neq b_j
					\end{cases}
				\]
			\State \hspace{16pt} Add $n ^ {\tau - \ell}$ copies of the row vector $w(u,v)$ \\
            \makebox[43pt]{} to the matrix $M$
		\State \textbf{end for}
	\EndFor
    \EndFunction
\end{algorithmic}
\end{algorithm}

\begin{proofof}{\Lemma{BuildOAWorks}}
Let us view $C$
as a multiset $\cH_q$ of functions mapping $[m]$ to $[q]$ (cf.~\Claim{Equiv}).
Since the dual of $C$ has minimum distance at least $t+1$, it follows that $\cH_q$ is \SU{t} \cite[Exercise 2.13]{GRS} 
\cite[Theorem 10.17]{StinsonBook}.

\textsc{BuildOA} produces a matrix $M$ with $m$ columns, where each entry is in $[n]$.
Let $s$ be the total number of rows added.
Since each vector $u \in C$ produces exactly $n^\tau$ rows in $M$, and $C$ has dimension $k$, we have
\[ s = n^\tau \card{C} = n^\tau q^k, \] as required.
By \Claim{Equiv} again, we will view $M$ as a multiset $\cH = \set{h_1,\ldots,h_s}$ of $[m]\rightarrow[n]$ functions. 
It remains to show that $\cH$ is \SU{t}.

Consider a function $h \in \cH_q$.
For a sequence $y \in [n] ^ m$, define the function $f' _ {h, y} \colon [m] \to [n]$ by
\[
	f'_{h, y}(j) = 
	\begin{cases}
		\phi_{b_j}(h(j)) 	& \text{ if } h(j) \neq b_j \\ 
		y_j 		& \text{ if } h(j) = b_j.
	\end{cases}
\]
Here the ``bad'' hash values are replaced with entries of $y$, and the other hash values are mapped by $\phi$ into $[n]$.
Let $\cH'$ be the multiset made by adding one copy of $f'_{h, y}$ (counting multiplicities) for each $y \in [n]^m$ and $h \in \cH_q$, i.e.,
\[
	\cH' = \setst{f' _ {h, y} }{ y \in [n]^m , ~ h \in \cH_p }.
\]
Note that each function $h \in \cH_q$ gives rise to exactly $n ^ m$ functions in $\cH'$, although they are not necessarily unique.
It is clear that $\cH'$ is the family made by repeating each function in $\cH$ exactly $n ^ {m - \tau}$ times.
Therefore $\cH'$ is \SU{t} if and only if $\cH$ is \SU{t}.

We will show that $\cH'$ is \SU{t} by showing that, for any fixed $i_1 , \ldots, i_t \in [m]$, the random variables $h'(i_1), \ldots, h'(i_t)$ are independent and uniformly distributed over $[n]$ when $h' \in \cH'$ is chosen uniformly at random.
Note that by definition of $\cH'$,
uniformly sampling a function $h'$ from $\cH'$ is (in distribution) equivalent to  uniformly sampling a function $h \in \cH_q$, and a sequence $N=(N_1, \ldots, N_m)$ of independent and uniformly distributed random variables in $[n]$, and then returning $h' = f'_{h,N}$.

To show that $h'(i_1),\ldots,h'(i_t)$ are $t$-wise independent\footnote{Here the term ``$t$-wise independent'' has the meaning from probability theory, where it does not imply that the random variables are uniform. We use the term ``$t$-independent'' to indicate the meaning in pseudorandomness, where the random variables must additionally be uniformly distributed.}, we will show that they are a deterministic function of random variables that are themselves $t$-wise independent.
That is, the pairs
\[ (h(i_1),N_{i_1}), \ldots, (h(i_t),N_{i_t}) \]
are easily seen to be $t$-wise independent.
We obtain the values $h'(i_1),\ldots,h'(i_t)$ by applying the function $f \colon \bF_ q \times [n] \times \bF_q \to [n]$, defined by 
\[
	f(x , r, \beta) = 
	\begin{cases}
		\phi_\beta(x)	& \text{ if } x \neq \beta \\
		r 		& \text{ if } x = \beta.
	\end{cases}
\]
Notice that by definition of $ f' _ {h, N}$, we have 
\[
	h'(i_j) = f(h(i_j) , N_{i_j}, b_{i_j}), \quad\forall j \in [t].
\]
Since $f$ and $b$ are deterministic, it follows that
$ h'(i_1),\ldots,h'(i_t) $
are $t$-wise independent.

To complete the proof that $\cH'$ is \SU{t}, we must also show that each $h'(i)$ is uniformly distributed.
To see this, note that for each $i \in [m]$ and $\nu \in [n]$
\begin{align*}
	\prob{h'(i) = \nu}
 &~=~ \prob{h(i) = b_i \land N _ i = \nu} + \prob{h(i) \in \phi_{b_i} \inv (\nu)} \\
 &~=~ \frac{1}{q} \cdot \frac{1}{n}
  + \frac{(q-1)}{nq}
   ~=~ \frac{1}{n}
\end{align*}
since $h(i)$ is uniformly distributed, and  $\phi$ is a $\frac{q-1}{n}$-to-$1$ map.

To summarize, for any $i_1, \ldots, i_t \in [m]$, $h'(i_1), \ldots, h'(i_t)$ are independent and uniformly distributed random variables over $[n]$.
That is, $\cH'$ is \SU{t}. 
As previously argued, this implies that $\cH$ is \SU{t}, which is what we intended to prove.
\end{proofof}

\begin{corollary}
\CorollaryName{BuildOAWorks}
Let $m,n,t$ be integers with $n \geq 2$ and $2 \leq t \leq m$.
Let $q$ be a prime power satisfying $q \equiv 1 \pmod n$.
Let $C_1 \subsetneq C_2 \subseteq \bF_q^m$ be linear codes such that the dual of $C_1$ has minimum distance at least $t+1$.
Let $k$ be the dimension of $C_1$
and $d_2$ be the minimum distance of $C_2$.
Then there is an algorithm to produce an 
$[s,m,n,t]$ orthogonal array with
$s = n^{m-d_2} q^k$.
\end{corollary}
\begin{proof}
Pick any $b \in C_2 \setminus C_1$.
Then $\Delta(b,u) \geq d_2$ for all $u \in C_1$.
Thus we may apply \Lemma{BuildOAWorks}, taking $C := C_1$ and $\tau := m - d_2$.
\end{proof}

\section{Construction using Reed-Solomon codes}
\SectionName{RS}

In this section, we instantiate the abstract bound using Reed-Solomon codes. 

\begin{theorem}
\TheoremName{RSOA}
There exists a universal constant $c>0$ such that the following is true.
For any integers $n , m , t \geq 2$ with $t \leq m $, there exists an explicit $[s,m,n,t]$ orthogonal array with
\[ s ~\leq~ n^t \cdot \paren{c(n + m)}^ {5t}. \]
Assuming the generalized Riemann hypothesis,
\[ s ~\leq~ n^t \big((n+m) \ln(n+m)\big)^{2t}. \]
\end{theorem}

This bound matches (up to constants in the exponent) the non-constructive bound in the conference article of Kuperberg, Lovett and Peled \cite{KLP12}.
This bound is not near-optimal (i.e., $(mn/t)^{O(t)}$) for all parameters, but it is in many regimes.
(For example, it is near-optimal if $t \leq \min\set{n,m}$, or $t \leq m^{0.99}$, or $n \geq m^{0.01}$. Also, if $t = \Omega(m)$, then the trivial bound of $s \leq n^m$ is near-optimal.)
\Section{AG} presents an improved construction that is near-optimal for \emph{all} parameters, matching the non-constructive results of the journal article \cite{KLP17}.

\paragraph{The prime.}
The first step of the proof is to find a prime $p$ satisfying $p \equiv 1 \pmod n$ and $p > m$.
We will use the following result.

\begin{theorem}[Linnik's theorem]
\TheoremName{Linnik}
There exist universal constants $L, c_L > 0$ such that the following is true.
For any integer $n \geq 2$, and any positive integer $a$ coprime with $n$,
there exists a prime $p$ satisfying
(i) $p \equiv a \pmod n$,
and (ii) $p \leq c_L \cdot n^L$.
\end{theorem}

The constant $L$ is known as Linnik's constant.
Although Linnik in 1944 did not provide an explicit value for $L$, the most recent developments.
have shown that $L \leq 5.5$ in 1992 \cite{HeathBrown}, $L < 5.2$ in 2011 \cite{Xylouris},
and $L \leq 5$ in \cite[Theorem 2.1]{XylourisPhD}.
Assuming the generalized Riemann hypothesis \cite{Lamzouri}, the conclusion (ii) can be strengthened to $p \leq (n \ln n)^2$.

\begin{corollary}
\CorollaryName{Linnik}
For any constants $L, c_L > 0$ satisfying Linnik's theorem, the following is true.
For any integer $n \geq 2$ and real $m \geq 2$, 
there exists a prime $p$ satisfying
(i) $p \equiv 1 \pmod n$,
(ii)~$p \leq c_L (n+m)^L$,
and (iii) $m < p$.
\end{corollary}
\begin{proof}
Let $\eta = n \cdot \ceil{m / n}$. Note that $m \leq \eta < m + n$. Let $p$ be the prime obtained by \Theorem{Linnik} with $\eta$ instead of $n$, and with $a=1$ (which is trivially coprime with $n$).
The theorem ensures that $\eta \divides p - 1$; since $n \divides \eta$, we have $n \divides p - 1$, which establishes (i).
Since $\eta \divides p - 1$ and $p-1>0$, we must have $m \leq \eta < p$, which establishes (iii).
Lastly, we have $p \leq c_L \cdot \eta^L < c_L \cdot (m + n)^L$, which establishes (ii).
\end{proof}

A prime $p$ satisfying the conditions of \Corollary{Linnik} can be found by exhaustive search. 
Simply test the primality of all integers $p$ satisfying $p \equiv 1 \pmod {\eta}$ and $p \leq c_L \cdot \eta^L$, of which there are at most $\eta^{L-1}$.
Since the primality of $x$ can be tested in $O(\log^7 x)$ time, the overall runtime is $O(\eta^{L-1} \log^7 \eta) = \tilde{O}((n+m)^{4})$.

\paragraph{The codes.}
The next step of the proof is to find codes of appropriate parameters that can be used with \Corollary{BuildOAWorks}.
Let $q$ equal the prime $p$ chosen above.
We will use Reed-Solomon codes in $\bF_q^m$.
The following theorem states their basic properties, with apparently excessive detail, in order to draw a parallel with \Theorem{Goppa} below.

\begin{theorem}
\TheoremName{RS}
Let $m \leq q$. There exists a sequence of linear codes $C_0, \ldots, C_{m-1} \subseteq \bF_q^m$, where $C_a$ has parameters $[m, k_a, d_a]_q$, such that the following statements hold.
\begin{alignat}{2}
C_a \subsetneq C_{a+1} \quad\forall ~ 0 \leq a < m-1 
\\
\EquationName{RSDist}
d_a = m-a \quad\forall ~ 0 \leq a < m
\\
\EquationName{RSDim}
k_a = a+1 \quad\forall ~ 0 \leq a < m
\\
\EquationName{RSDual}
k_a ^ \perp + d_a ^ \perp = m  + 1 \quad~ \forall ~ 0 \leq a < m 
\end{alignat}
Here $k_a^\perp$ and $d_a ^ \perp$ respectively denote the dimension and distance of $C_a^\perp$, the dual of $C_a$.
\end{theorem}

Proofs of the claims in this theorem may be found in standard references, e.g.,~\cite[Claim 5.2.1]{GRS}
\cite[Proposition 11.4]{Roth} \cite[Exercise 1.1.20]{TVN}.

We will apply \Corollary{BuildOAWorks} to the codes $C_{t-1} \subsetneq C_{t}$.
It follows from \eqref{eq:RSDual} that
\[ m+1 ~=~ k_{t-1}^\perp+d_{t-1}^\perp ~=~ (m-k_{t-1}) + d_{t-1}^\perp,
\]
and therefore
\[
d_{t-1}^\perp = k_{t-1}+1 = t+1.
\]
Thus \Corollary{BuildOAWorks}
yields an $[s,m,n,t]$ orthogonal array with
\begin{equation}
\EquationName{RSSize}
s ~=~ n^{m-d_t} p^{k_{t-1}}
 ~=~ n^t p^t
 ~\leq~ n^{t} \big(c_L (n+m)^{5} \big)^{t}.
\end{equation}
This proves \Theorem{RSOA}.

\section{Construction using algebraic geometry codes}
\SectionName{AG}

In this section we replace the Reed-Solomon codes with algebraic geometry codes, which gives a near-optimal orthogonal array for all parameters.
This matches the non-constructive bound shown in the journal article~\cite{KLP17}, and proves \Theorem{Main}.

\begin{theorem}
\TheoremName{AGOA}
There exists a universal constant $c>0$ such that the following is true.
For any integers $n , m , t \geq 2$ with $t \leq m$, there exists an $[s,m,n,t]$ orthogonal array with
\[ s ~\leq~ n^{5t} \big(c(n+11+m/t)^{5} \big)^{6t}. \]
Moreover, this array can be constructed
by a deterministic algorithm in $\poly(sm)$ time.
\end{theorem}

To motivate the proof, let us reflect on the construction of the previous section.
A standard bound on all $[m,k,d]$ linear codes is the \textit{Singleton bound},
which states that $k + d \leq m + 1$;
if equality holds it is called an \textit{MDS code}.
In \Section{RS}, to apply \Corollary{BuildOAWorks} we needed $d_{t-1}^\perp \geq t+1$.
Furthermore, to obtain an exponent of $O(t)$ in \eqref{eq:RSSize},
we wanted $d_t \geq m - O(t)$ and $k_{t-1}=O(t)$.
Since Reed-Solomon codes and their duals are MDS codes, we have
exactly $d_{t-1}^\perp=t+1$, $d_t = m-t$ and $k_{t-1}=t$.
This yields an orthogonal array with $s = n^t p^t$ rows,
with the main shortcoming being the undesirably large field size of $p = \poly(n+m)$.

The construction of this section attains an improved bound by reducing the field size to $\poly(n+m/t)$.
We cannot hope to use MDS codes anymore, since it is believed that they do not exist when the field size is less than $m-1$.
Instead, we will use codes that only approximately satisfy the Singleton bound, but do so over a much smaller field.
Algebraic geometry codes (AG codes) provide a tradeoff suitable for our purposes.

\paragraph{Algebraic geometry codes.}
These are a general class of codes that involve evaluations of rational functions over algebraic curves \cite{Goppa}. 
Reed-Solomon codes are a simple special case using evaluations of polynomials.
We begin with a detailed statement of the codes that can be constructed from a general curve.

Let $q$ be a prime power.
Let $X$ be a curve over $\bF_q$ (see \cite[\S 2.1.2]{TVN}).
Let $g$ be its genus (see \cite[pages 87 and 91]{TVN}).
Let $N$ be its number of rational points (see \cite[pp.~127]{TVN},  \cite[Definition VII.6.14]{Lorenzini}, or \cite[Section A.2]{NX}).

\begin{theorem}
\TheoremName{Goppa}
Let $m < N$. There exists a sequence of linear codes $C_0, \ldots, C_{m-1} \subseteq \bF_q^m$, where $C_a$ has parameters $[m, k_a, d_a]_q$, such that the following statements hold.
\begin{align}
\begin{array}{r}
C_a \subseteq C_{a+1} \quad\forall ~ 0 \leq a < m-1 
\\
\footnotesize \text{\normalfont \cite[Remark 4.1.7]{TVN}}
\end{array}
\\
\EquationName{GoppaDist}
\begin{array}{r}
d_a \geq m-a \quad\forall ~ 0 \leq a < m
\\
\footnotesize \text{\normalfont \cite[Theorem 4.1.1]{TVN}}
\end{array}
\\
\EquationName{GoppaDim}
\begin{array}{r}
k_a \geq a-g+1 \quad\forall ~ 0 \leq a < m, \\
\text{\normalfont and equality occurs for $a \geq 2g-1$} \\
\footnotesize \text{\normalfont \cite[Remark 4.1.4]{TVN}}
\end{array} \\
\EquationName{GoppaDual}
\begin{array}{r}
k_a ^ \perp + d_a ^ \perp \geq m -g + 1 \quad~~ \forall ~ 0 \leq a < m \\
\footnotesize \text{\normalfont \cite[Exercise 4.1.27]{TVN}}
\end{array}
\end{align}

\end{theorem}

Naively one might imagine that $C_a \subsetneq C_{a+1}$ for all $a$.
However that is not necessarily the case since \eqref{eq:GoppaDim} does not exactly specify the dimension: for small $a$ it only provides a lower bound.
However, this will not be problematic for our purposes.
We will only consider $C_a$ for $a \geq 2g-1$, in which case \eqref{eq:GoppaDim} determines the dimension exactly.

Clearly, for a curve of genus $0$, each code $C_a$ must be an MDS code,
since adding \eqref{eq:GoppaDist} and \eqref{eq:GoppaDim} shows that the Singleton bound is tight.
More generally, the AG code construction is most efficient when it is based on a curve for which $g/N$ is as small as possible.
However this ratio cannot be arbitrarily small.
It is known\footnote{This follows from the Drinfeld-Vl\u{a}du\c{t} theorem,
but is weaker since we do not need exact constants.}
that $g/N \geq c/\sqrt{q}$ for some constant $c>0$ and sufficiently large $g$.
In fact, this bound is asymptotically tight.
Various explicit curves asymptotically matching this bound have been discovered by Tsfasman, Vl\u{a}du\c{t} and Zink~\cite{tsfasman1982modular}, Garcia and Stichtenoth~\cite{garcia1995tower}, and others.

Using such curves, and letting $m = \Omega(N)$, one may see from \eqref{eq:GoppaDist} and \eqref{eq:GoppaDim} that the Singleton bound nearly holds, but there is a ``defect'' of $g=O(m/\sqrt{q})$.
The rough idea of our analysis is to choose parameters such that $\sqrt{q} = \Theta(m/t)$, so that the defect is only $O(t)$, which will be acceptable for our purposes.

There is a slight complication.
The use of algebraic curves in coding theory primarily concerns limiting behaviour of the rate/distance tradeoff as the code length $m$ tends to infinity.
Since we wish to construct orthogonal arrays for \emph{all} parameters $m,n,t$, we must ensure that curves exist with values of $N$ that are sufficiently dense. 
It is not necessarily true that all curve families that have previously been used with AG codes will be suitable for our purposes.
Fortunately, the modular curves discussed in the following theorem are suitably dense.

\begin{theorem}[{\protect \cite{KTV}, or \cite[page~18]{TVN2}}]
\TheoremName{KTV}
Let $p \neq \ell$ be distinct primes with $p \neq 11$. Then the classical modular curve $X_0(11\ell)$ over $\bF_{p^2}$ has at least $N= (p-1)(\ell+1)$ rational points and has genus $g = \ell$.
Furthermore, the generator matrix of the codes made from such curves can be constructed in time $\poly(m)$.
\end{theorem}

We remark that the celebrated work of \citep{tsfasman1982modular} also used the classical modular curve $X_0(\cdot)$ over $\bF_{p^2}$, but they did not require the genus to be a multiple of $11$.

\paragraph{Using AG codes for orthogonal arrays.}
Next we explain how the AG codes can be used to construct orthogonal arrays.
Let $p$ be the smallest prime satisfying
\[ 
p \equiv 1 \pmod n
\qquad\text{and}\qquad
p \geq 11 + m/t.\]
Clearly $p \neq 11$. By Linnik's theorem (\Corollary{Linnik}), $p \leq c_L (n+11+m/t)^{5}$.
Let $q = p^2$ and note that $q \equiv 1 \pmod n$.

We choose $\ell$ to be the smallest prime such that
$\ell \geq t$;
by Bertrand's postulate $\ell \leq 2t$.
By \Theorem{KTV}, the classical modular curve $X_0(11\ell)$ over $\bF_q$ satisfies 
\begin{equation}
\EquationName{ClassicalModular}
N = {(p-1)(\ell+1)} ~>~ m \qquad \text{and}\qquad g = \ell \leq 2t.
\end{equation}
Given this curve, \Theorem{Goppa} guarantees existence of a sequence of codes with various parameters.
We will only use the codes $C_u$ and $C_{u+1}$, where
\begin{equation}
\EquationName{UDef}
u ~=~ 2\ell - 1 + t ~=~ 2g -1 + t.
\end{equation} 
The following claim performs some simple calculations in preparation for using \Corollary{BuildOAWorks}.

\begin{claim} \mbox{}

\begin{enumerate}
\item $C_u \subsetneq C_{u+1}$.
\item $d_u^\perp \geq t+1$.
\item $k_u \leq 3t$.
\item $d_{u+1} \geq m - 5t$.
\end{enumerate}
\end{claim}
\begin{proof}\mbox{}

\begin{enumerate}
\item By \eqref{eq:UDef} we have $u > 2g-1$, so \eqref{eq:GoppaDim} holds with equality for $a \in \set{u,u+1}$, and therefore $k_u < k_{u+1}$.

\item $d_u^\perp
 \geq m - k_u^\perp - g + 1
 = k_u - g + 1
 \geq u - 2g + 2
 = t+1,$
by \eqref{eq:GoppaDual},
\eqref{eq:GoppaDim}, and \eqref{eq:UDef}.

\item By \eqref{eq:UDef} we have $u > 2g-1$, so \eqref{eq:GoppaDim} holds with equality for $a=u$. Thus
\[
k_u ~=~ u-g+1
 ~=~ (2g-1+t)-g+1
 ~\leq 3t,
\]
by \eqref{eq:GoppaDim}, \eqref{eq:UDef}, and \eqref{eq:ClassicalModular}.

\item $d_{u+1}
    \geq m-u-1
    = m-(2g-1+t)-1
    \geq m-5t, $
by
    \eqref{eq:GoppaDist},
    \eqref{eq:UDef}, and
   \eqref{eq:ClassicalModular}.
   \qedhere
\end{enumerate}
\end{proof}

We apply \Corollary{BuildOAWorks} to $C_u$ and $C_{u+1}$ with $k \leq 3t$
and $d_2 \geq m-5t$.
Thus, there is an $[s,m,n,t]$ orthogonal array with 
\[
s ~\leq~ n^{5t} q^{3t}
 ~\leq~ n^{5t} \big(c(n+11+m/t)^{5} \big)^{6t}.
\]

Regarding the algorithmic efficiency, \Theorem{KTV} above states that the generator matrix for $C_u$ and $C_{u+1}$ can be constructed in time $\poly(m)$.
Given those matrices, a vector $v \in C_{u+1} \setminus C_u$ (as required by \Corollary{BuildOAWorks}) can be found in $\poly(sm)$ time. Since \Algorithm{BuildOA2} also runs in $\poly(sm)$, it follows that the construction of \Corollary{BuildOAWorks} runs in $\poly(sm)$.
This concludes the proof of \Theorem{AGOA}.

\section{Implementing a \SU{t} hash function}
\SectionName{Hash}

Although orthogonal arrays and \SU{t} hash functions are mathematically equivalent (see \Claim{Equiv}), a key difference is the efficiency of their implementations.
A hash function needs to be represented in little space, and the algorithms for constructing and evaluating it should be fast. 

In \Algorithm{hash_alg}, we present pseudocode for a \SU{t} hash function based on the orthogonal array construction of Sections~\ref{sec:GenericConstruction} and \ref{sec:RS}.
We assume that $\Uniform(S)$ returns a uniformly random element of the finite set $S$ in $O(1)$ time.

\begin{algorithm}[t!]
\caption{Pseudocode for a \SU{t} hash function.
This pseudocode involves the constant $\nu$ (from \Theorem{q_linnik}), which is effectively computable (see Remark~\ref{rem:comput}).}
\AlgorithmName{hash_alg}
\block{
class \HashNoSpace:
}{
    dictionary $D$ \\
    $\rhd$ $D$ has size at most $t$, so it can be implemented as an array\\
    int $n,m,t,p$ \\
    int $a[0..t-1]$ \\
    \\
    \Block{\CnstrNoSpace(int $n, ~m, ~t$)}
    {
        Store parameters $n, m, t$ in the object \\
        Let $\eta \gets n \cdot \ceil{m/n}$\\
        Find a prime $p \in [\eta+1, \eta^\nu]$ such that $p \equiv 1 \pmod {\eta}$ \\
        $\rhd$ Generate coefficients of a random polynomial \\
        \mbox{}\hspace{6pt} of degree-$(t\!-\!1)$ \\
        \Block{\MyFor $i \gets 0 \ldots t-1$}
        {
            $a_i \gets \Uniform(\bF _ p)$
        } \\
        $D \leftarrow$ empty dictionary
    } \\
    \\
    $\rhd$ Evaluate the hash function at input $x \in [m]$ \\
    \Block{\EvalNoSpace(int $x$)}
    {
        \Block{\MyIf $x \in \mathrm{keys}(D)$}
        {
            $\rhd$ The hash value of $x$ is already in the dictionary \\
            \MyReturn $D[x]$
        } \\
        $\rhd$ Evaluate the chosen polynomial at $x$ (over $\bF_p$) \\
        Let $y \gets \sum_{i=0}^{t-1} a _ i x ^ i$ \\
        \Block{\MyIf $y=x^t$
            \quad$\rhd$ Here $x^t$ is the ``bad value'' $b(x)$}
        {
            $\rhd$ Choose a replacement value randomly \\
            Let $U \gets \Uniform([n])$ \\
            $\rhd$ Store the hash value of $x$ in the dictionary \\
            $D[x] \gets U$ \\
		\MyReturn $U$
        } \\
        $\rhd$ Return $\phi_{b(x)}(y)$ \\
        \MyReturn $1+(((y+p-1-x^t) \text{~mod~} p) \text{~mod~} n)$
    }
}
\end{algorithm}

\begin{theorem}
\TheoremName{BasicHash}
In the pseudocode of \Algorithm{hash_alg}:
\begin{itemize}
    \item \Cnstr has expected runtime $O\big(t+(n+m)^\eps\big)$ for every $\eps>0$.
    \item \Eval has runtime $O(t)$.
    \item The space for a \Hash object is $O(t \log(n+m))$ bits.
\end{itemize} 
Assuming the generalized Riemann hypothesis, the time for \Cnstr is $O(t + \log(n+m))$.
\end{theorem}

\newcommand{\NumPrimes}[3]{\pi(#1;#2,#3)}

The proof involves the following notation,
for $x \in \bR$, $\eta, a \in \Z$, $\eta \geq 2$,
\begin{align*}
S_{x,\eta,a}
 &~=~ \setst{ i \in \bN}{i \equiv a \Pmod \eta ,\: 
  1 < i \leq x } \\
\NumPrimes{x}{\eta}{a}
 &~=~ \card{\setst{p \in S_{x,\eta,a} }{ \text{\normalfont $p$ is prime}}}
\end{align*} 
as well as the Euler totient function $\phi(\eta)$, which counts the number of positive integers up to $\eta$ that are relatively prime to $\eta$.
Throughout this section, we assume that $a$ and $\eta$ are coprime.

\begin{proof}
It is not difficult to see that this pseudocode implements a \SU{t} hash function, since it implements the construction of Sections~\ref{sec:GenericConstruction} and \ref{sec:RS}. We now summarize the main ideas.

The key task of the \Cnstr is to find a prime $p$ satisfying $p \equiv 1 \pmod n$ and $\eta+1 \leq p \leq \eta^\nu$, for some constant $\nu$ defined in \Theorem{q_linnik}.
The runtime of that step is discussed extensively below.
It then randomly generates the coefficients of a uniformly random degree-$(t\!-\!1)$ polynomial, which is analogous to uniformly selecting the function $h \in \cH_q$ in the proof of \Lemma{BuildOAWorks}. (Here $q=p$.)
The ``bad vector'' $b$ is defined to be the evaluations of the degree-$t$ monomial, i.e., $b(x) = x^t$. This is a codeword in  $C_{t}$ but not in $C_{t-1}$, and so it can be used as in the proof of \Theorem{RSOA}.   
The \Eval function calculates $h(x)$, and checks whether it equals the forbidden value $b(x)$. If so, it replaces it with a random value and caches that in the dictionary $D$. Otherwise it returns $\phi_{b(x)}(h(x))$.

We claim that there are at most $t$ inputs $x \in [m]$ for which $\sum_{i=0}^{t-1} a_i x^i = x^t$
(i.e., $h(x)=b(x)$).
This follows since the polynomial $h(x) - b(x)$ is of degree $t$ and hence has at most $t$ roots.
This implies that the dictionary $D$ will have size at most $t$.
Each key in the dictionary is in $[m]$ and each value is in $[n]$, so the space required for $D$ is $O(t \log(n+m))$ bits.
Since $p \leq (n+m)^{\nu}$, and each $a_i \in \bF_p$, the space required for all other parameters of the \Hash object is $O(t \log(n+m))$ bits.

\Eval can clearly be implemented in $O(t)$ time, even if the dictionary $D$ is implemented as an unsorted array.

The main challenge of the analysis is the time required to find the prime $p$. 
In \Section{RS} the time to find the prime $p$ was negligible compared to the time to build the orthogonal array, so a (deterministic) algorithm with runtime $O\big((n+m)^{4}\big)$ was sufficient.
To obtain an improved runtime for \Cnstr, we will use random sampling.
\Theorem{q_linnik} below shows that there exists a constant $\nu$ such that,
assuming $x \geq \eta^\nu$, we have
\[
 \pi(x; \eta, a)
 ~\geq~ \frac{c_\eps x}{\phi(\eta) \ln(x) \eta^\eps}
 ~\geq~ \frac{c_\eps x}{\eta^{1+\eps} \ln x},
\]
using the trivial bound $\phi(\eta) \leq \eta$.

Let us now fix $x = \eta^\nu$ and $a=1$.
The algorithm repeatedly samples integers at random from the set $S_{x,\eta,a}$, testing each for primality.
The expected number of iterations until finding a prime is
\begin{equation}
\EquationName{NumIters}
\frac{\card{S_{x,\eta,a}}}{\NumPrimes{x}{\eta}{a}}
~\leq~
\frac{x/\eta}{c_\eps x/(\eta^{1+\eps} \ln x)}
~=~  \frac{\eta^{\eps} \ln x}{c_\eps},
\end{equation}
since $x \geq \eta$.
Each primality test requires $O(\polylog(x))$ time.
Using that $\eta \leq (n+m)$, we have shown that the expected time to find the prime $p$ is $O((n+m)^\eps)$ for every $\epsilon > 0$.

Assuming the generalized Riemann hypothesis, \Corollary{IK} yields
\[
\NumPrimes{x}{\eta}{a} ~\geq~ \frac{x}{2 \eta \ln x},
\]
so the expected number of iterations until finding the prime $p$ is 
\[
\frac{\card{S_{x,\eta,a}}}{\NumPrimes{x}{\eta}{a}}
~\leq~ \frac{2x/\eta}{x/\eta \ln x}
~=~ 2 \ln(x),
\]
which is $O(\log(n+m))$ since $x=\eta^\nu$.
\end{proof}

The next theorem follows from known results in the analytic number theory literature, as we explain below.

\begin{theorem}[Quantitative Linnik's Theorem]
\TheoremName{q_linnik}
There exists a universal constant $\nu \geq 1$ (independent of $\eta, \ep$) and a constant $c_\ep > 0$ such that, for any $\ep > 0$, any $x \geq \eta^\nu$,
\[
	\NumPrimes{x}{\eta}{a}
	~\geq~
	c_\ep \paren{ \frac x {\phi(\eta) \ln x} \cdot \frac 1 {\eta ^ {\ep}} }
 .
\]
\end{theorem}

The distribution of the primes in an arithmetic progression is closely tied to the well-studied $L$-functions of characters mod $\eta$. 
As with many results in analytic number theory, there are different cases depending on whether the $L$-functions have so-called ``exceptional zeros''.
Whether or not these zeros exist is unknown, and relates to the generalized Riemann hypothesis.
We will use the following result, which is a corollary of \cite[Theorems 18.6, 18.7]{IK04}.

\begin{corollary}
\CorollaryName{IK}
There exist universal constants $\nu, c_1, c_2, c_3$ (independent of $x, \eta, a$) such that, if $x \geq \eta^\nu$ and $\eta \geq c_3$, then the following holds.
\begin{itemize}
\item If there exists a real character $\chi \pmod \eta$ whose $L$-function has a real zero $\beta$ satisfying $1 - \beta \leq c_1 / 2 \log \eta$ then 
\[
	\NumPrimes{x}{\eta}{a} ~\geq~ c_2 \cdot \frac {x} {\phi(\eta)} \cdot (1-\beta).
\]
\item If there exists no such character $\chi$, then 
\[
	\NumPrimes{x}{\eta}{a}
	~\geq~
	\frac{1}{2} \cdot \frac x {\phi(\eta)\ln x}.
\] 
\end{itemize}
The generalized Riemann hypothesis implies that the second case holds.
\end{corollary} 

The results in \cite{IK04} are stated in terms of the von Mangoldt function $\Lambda$ and its sum $\psi$:
\begin{align*}
\Lambda(n) &~=~ \begin{cases} \ln p &\quad\text{if $n=p^k$ for a prime $p$ and integer $k \geq 1$} \\
0 &\quad\text{otherwise}.
\end{cases} \\
\psi(x;\eta,a) &~=~ \sum_{\substack{n \leq x  \\ n \equiv a \Pmod \eta}} \Lambda(n).
\end{align*}
\Corollary{IK} follows from those results using the simple inequality $\psi(x;\eta,a) \leq \NumPrimes{x}{\eta}{a} \cdot \ln(x)$.

Let us now return to the proof of \Theorem{q_linnik}.
If the second case of \Corollary{IK} applies, then \Theorem{q_linnik} is immediate.
If the first case of \Corollary{IK} applies, then we must ensure that the exceptional zero is significantly less than $1$.
The following is a classic result of Siegel \cite{Sig35} \cite[Corollary 11.15]{MV06}.

\begin{theorem}[Siegel, 1935]
\TheoremName{Siegel}
For every $\eps > 0$ there exists a constant $c_\eps$ such that, if $\chi$ is a real character$\Pmod \eta$ with $L(s , \chi)$ having a real zero $\beta$ then 
\[
	\beta ~<~ 1 - c_\eps \eta ^ {- \eps}.
\]
\end{theorem}

Thus, in the first case of \Corollary{IK}, it holds that $1-\beta \geq c_\eps \eta^{-\eps}$, which completes the proof of \Theorem{q_linnik}.

\begin{remark}
\label{rem:comput}
The numerical value of constants $\nu, c_L, L, c_1, c_2, c_3$ satisfying the above properties can be explicitly computed ``given the time and will''; see \cite[Pages 427, 428 and property 18.90]{IK04}. Such constants are called ``effectively computable'' in the number theory literature. In contrast, the value of $c_\eps$ is ``ineffective'' for $\eps < \frac 1 2$ \cite[Theorem 5.28]{IK04}, meaning that the proof of the result does not yield a way to compute the constant.
However, note that implementing Algorithms~\ref{alg:BuildOA2} and \ref{alg:hash_alg} only requires knowing the value of the constants $\nu, c_L, L$.
The constant $c_\eps$ only appears in the analysis of \Theorem{BasicHash}. Specifically, from \eqref{eq:NumIters}, it appears inside the $O(.)$ bound on the runtime of \Cnstr.
\end{remark}

\section{Construction using random linear codes}
\SectionName{GV}

Sections~\ref{sec:RS} and \ref{sec:AG} use deterministic code constructions.
In this section we show that a random construction of linear codes can also be used to prove an non-explicit form of \Theorem{Main}, using an efficient, randomized algorithm.

\begin{lemma}
\LemmaName{MainOptimal}
There exists a universal constant $c > 0$ such that the following is true.
For any $m , n , t \in \N$, there exists a prime number $p$, a linear code $C \subseteq \bF_p^m$ and a function $b \colon [m] \to \bF _ p$ such that 
\begin{itemize}
        \item $p \equiv 1 \pmod n$
        and $p \leq c_L \big(n+(me/t)^3)^L$.
	\item the dimension of $C$ is at most $2t$.
        \item the distance of $C^\perp$ is at least $t+1$.
	\item $\Delta(b,v) \geq m-3t ~\forall v \in C$.
\end{itemize}
Moreover, there exists a randomized algorithm that
outputs $p$, $C$ and $b$ satisfying these conditions
and has expected runtime $(mn/t)^{O(t)}$.
\end{lemma}

The code $C$ and vector $b$ are then provided as input to \Algorithm{BuildOA2}. 
Applying \Lemma{BuildOAWorks}, we obtain an $[s,m,n,t]$ orthogonal array with $s \leq n^{3t} p^{2t} = (mn/t)^{O(t)}$.
The expected runtime is $(mn/t)^{O(t)}$.

The proof of \Lemma{MainOptimal} requires the following lemma. Here, the matrix $M$ is simply the generator matrix for the dual of a code meeting the Gilbert–Varshamov bound.
A non-algorithmic form of this lemma may be found in
\cite[Theorem 10.11]{StinsonBook}.

\begin{lemma}
\LemmaName{AlgGV}
Let $m, t \in \bN$ and let $p$ be a prime power.
Assume that
\begin{equation}
\EquationName{GVCondition}
\sum_{i=1}^t \binom{m}{i} (p-1)^i ~\leq~ p^{\ell}/4.
\end{equation}
Then there exists a randomized algorithm that, with probability at least $3/4$, outputs an $\ell \times m$ matrix $M$ such that any $t$ columns of $M$ are linearly independent.
\end{lemma}

The proof is classical, but for completeness we include it below.

\begin{proofof}{\Lemma{MainOptimal}}
Define
\[ \ell = 2t \qquad\text{and}\qquad s = 3t. \]
Let $p$ be a prime satisfying both
\[
	p \equiv 1 \pmod n \qquad\text{ and }\qquad \paren{me/t} ^ 3 < p. \label{p_require}
\]
By \Corollary{Linnik}, there exists such a $p$ with 
\begin{equation}
\EquationName{BoundOnP}
p ~\leq~ c_L \big(n + (me/t)^3 \big)^L.
\end{equation}
We can find $p$ deterministically with runtime $(cnm/t)^c$
for some constant $c$.

Next we claim that \eqref{eq:GVCondition} is satisfied.
Using $t \leq (m-1)/2$ and $p > 2$, the LHS is at most 
$t \binom{m}{t} p^t$.
Dividing both sides by $p^\ell$, we get
\begin{align*}
    t \binom{m}{t} p^{t-\ell}
    &~\leq~ t (m e / t)^t p^{-t}
        \qquad\text{(using $\ell = 2t$)} \\
    &~<~ t (m e/t)^{-2t}
        \qquad\text{(using the lower bound on $p$)} \\
    &~<~ t e^{-4t} ~\leq~ 1/4e,
\end{align*}
using that $m \geq e \cdot t$.
Since this is less than $1/4$, \eqref{eq:GVCondition} is satisfied.

We will repeatedly use \Lemma{AlgGV} to generate a matrix $M$
until all subsets of $t$ columns are linearly independent.
This takes time $O(\binom{m}{t} (mt)^c)$ for some constant $c$.
Let $C$ be the list of vectors in $\bF_p^m$ obtained by taking all linear combinations of the rows of $M$.
Then $C$ is a linear code of dimension at most $\ell$ and dual distance at least $t+1$;
see, e.g., \cite[Theorem 10.4]{StinsonBook}.
The time to construct $C$ is $O(\card{C} m \ell)$.

Lastly, we will repeatedly generate $b \in \bF_p^m$ uniformly at random until $\Delta(b,v) \geq m-s ~\forall v \in C$.
For any $I \in \binom{[m]}{s}$, let $E_I$ be the event that
there exists $v$ for which $v_i=b_i ~\forall i \in I$.
We will show that $\prob{ \union_I E_I } < 1$.
Clearly
\begin{align*}
\prob{\cE_I}
    ~\leq~ \card{C} \cdot p^{-s}
    ~\leq~ p^{\ell-s}
    ~=~ p^{-t}.
\end{align*}
Thus, by a union bound,
\begin{align*}
\prob{ \union_I E_I }
    &~\leq~ \binom{m}{s} p^{-t} \\
    &~\leq~ (m e / s)^s p^{-t} \\
    &~=~ 3^{-3t} (m e / t)^{3t} p^{-t} \\
    &~<~ 3^{-3t} ~<~ 1,
\end{align*}
by the definition $s = 3t$ and the lower bound $p > (m e / t)^3$.
Thus, the expected number of trials until generating $b$ is a constant. Each trial can be executed in time $\card{C} m$.
Overall, the expected runtime is $(mn/t)^{O(t)}$.
\end{proofof}

\begin{proofof}{\Lemma{AlgGV}}
Let $M$ be a uniformly random matrix of size $\ell \times m$.
Let $V \subseteq \bF_p^m$ be the non-zero vectors of support size at most $t$.
For $v \in V$, let $E_v$ be the event that $M v=0$.
Then $\prob{E_v} = p^{-\ell}$, so
\[ \prob{ \Union_{v \in V} E_v }
 ~\leq~ \card{V} p^{-\ell}
 ~=~ \sum_{i=1}^t \binom{m}{i} (p-1)^i p^{-\ell}. \]
By the lemma's hypothesis, this is at most $1/4$.
If this event does not occur then every linear dependence among
the columns of $M$ involves more than $t$ columns.
\end{proofof}

\section{Open Questions}

There are several open questions related to this work.

\begin{enumerate}
\item The orthogonal array constructed by \Theorem{Main} will have duplicate rows, in general.
(Similarly, the hash family $\cH$ constructed by \Theorem{MainHash} will, in general, be a multiset.)
Can it be modified to eliminate the duplicate rows?

\item In \Theorem{BasicHash}, the expected runtime of  \Cnstr is $O\big(t+(n+m)^\eps\big) ~\forall\eps>0$.
Can it be improved to $O(t)+\polylog(n+m)$, without assuming the generalized Riemann hypothesis?

\item The hash function construction of \Section{Hash} was based on the Reed-Solomon construction of \Section{RS}.
Can the construction of \Section{AG} be used instead?
Although there are algorithms to construct the generator matrix in time $\poly(m)$ (see \Theorem{KTV}),
we would like to evaluate the hash function in time $\poly(t)$ while using only space $O(t \log(n+m))$.
At present we are unaware of an AG code construction that can be used for this purpose.

\item Although there are explicitly known values of $L$ for which Linnik's theorem (\Theorem{Linnik}) holds, at present it seems that there is no explicitly known pair $(L,c_L)$ such that the theorem holds.
\end{enumerate}

\section*{Acknowledgements}

We thank Jan Vondrak for suggesting this problem to us and for some initial discussions. We thank Lior Silberman and Greg Martin for very helpful discussions on analytic number theory.

\bibliographystyle{plain} 
\bibliography{OA}

\end{document}